\documentclass[journal,10pt,twocolumn,twoside]{IEEEtran}

\usepackage{cite}

\ifCLASSINFOpdf
   \usepackage[pdftex]{graphicx}
\else

\fi

\usepackage[cmex10]{amsmath}
\usepackage{algorithm}
\usepackage{algcompatible}
\usepackage{array}
\usepackage[caption=false,font=normalsize,labelfont=sf,textfont=sf]{subfig}
\usepackage{url}
\usepackage{xcolor}

\newcommand\NoThen{\renewcommand\algorithmicthen{}}

\definecolor{Darkgray}{gray}{0}

\usepackage{multicol,multirow}

\usepackage{amsthm}
\usepackage{makecell}

\newtheorem{proposition}{\bf Proposition}

\newcommand{\ve}[1]{\boldsymbol{#1}}

\newcommand{\argmax}{\operatornamewithlimits{argmax}}
\newcommand{\argmin}{\operatornamewithlimits{argmin}}

\newcolumntype{I}{!{\vrule width 1.2pt}}
\makeatletter
\def\hlinewd#1{%
\noalign{\ifnum0=`}\fi\hrule \@height #1 %
\futurelet\reserved@a\@xhline}

\newenvironment{rcases}
  {\left.\begin{aligned}}
  {\end{aligned}\right\rbrace}
  
\IEEEoverridecommandlockouts
\usepackage{fancyhdr}
\makeatletter
\let\ps@plain\ps@fancy
\makeatother

\begin{document}

\title{Competitive Energy Trading Framework for Demand-side Management in Neighborhood Area Networks}
\author{Chathurika~P.~Mediwaththe,~\IEEEmembership{Student~Member,~IEEE,}
        Edward~R.~Stephens,~David~B.~Smith,~\IEEEmembership{Member,~IEEE,}
        and~Anirban~Mahanti
\thanks{C. P. Mediwaththe, E. R. Stephens, D. B. Smith and A. Mahanti are with Data61 (NICTA), CSIRO,
Eveleigh, NSW 2015, Australia. e-mail: (m.mediwaththe@student.unsw.edu.au, stephens.ted1@gmail.com, David.Smith@data61.csiro.au, Anirban.Mahanti@gmail.com).}%
\thanks{C. P. Mediwaththe, E. R. Stephens and A. Mahanti are also with the University of New South Wales, Australia.}
\thanks{D. B. Smith is also with the Australian National University, Australia.}%
}

\maketitle
\begin{abstract}

This paper, by comparing three potential energy trading systems, studies the feasibility of integrating a community energy storage (CES) device with consumer-owned photovoltaic (PV) systems for demand-side management of a residential neighborhood area network. We consider a fully-competitive CES operator in a non-cooperative Stackelberg game, a benevolent CES operator that has socially favorable regulations with competitive users, and a centralized cooperative CES operator that minimizes the total community energy cost. The former two game-theoretic systems consider that the CES operator first maximizes their revenue by setting a price signal and trading energy with the grid. Then the users with PV panels play a non-cooperative repeated game following the actions of the CES operator to trade energy with the CES device and the grid to minimize energy costs. The centralized CES operator cooperates with the users to minimize the total community energy cost without appropriate incentives. The non-cooperative Stackelberg game with the fully-competitive CES operator has a unique Stackelberg equilibrium at which the CES operator maximizes revenue and users obtain unique Pareto-optimal Nash equilibrium CES energy trading strategies. Extensive simulations show that the fully-competitive CES model gives the best trade-off of operating environment between the CES operator and the users.

\end{abstract}

\begin{IEEEkeywords}
 Community energy storage, demand-side management, game theory, neighborhood area network.
\end{IEEEkeywords}

\IEEEpeerreviewmaketitle

\section{Introduction}
Smart grid developments facilitate reliable and economical demand-side management for leveling peak energy demands and reducing energy costs \cite{MohensianGT}. Small-scale demand-side management as in residential gated communities has received attention with the increasing popularity and cost reductions of household-distributed renewable power generation and storage technologies. Community energy storage (CES) devices can be integrated with novel small-scale demand-side management approaches to efficiently utilize onsite energy generation from consumer-owned renewable power resources such as rooftop photovoltaic (PV) systems \cite{ZHU}. These methods can create value for end users by reducing energy costs without modifying their electricity demand patterns \cite{Albadi}. In the context of the future energy grid, small-scale demand-side management with CES devices would play a vital role with the rapid growth of solar mini-grids such as Melbourne's Codstream Network \cite{minigrid}.

Small-scale demand-side management with a CES device and behind-the-meter PV systems requires a feasible framework with suitable incentives for all players. One potential approach is to devise the load management optimization that is controlled by a centralized entity. However, such centralized strategies may inflate costs for individual players. Moreover, robust operation of such a system would require system-wide information available to the controller including users' energy information that may increase the communication overhead \cite{smallScale}. In addition, the residential users may not subscribe to such a demand-side management approach as the central entity controls their personal energy decisions. A decentralized framework that can distribute the energy decision-making to individuals would be an effective alternative to overcome the above challenges.

In this paper, three energy trading systems with different CES operator structures are compared: a fully-competitive CES operator in a non-cooperative Stackelberg game, a benevolent CES operator that has socially favorable regulations with competitive consumers, and a centralized cooperative CES operator that collaborates with the users to minimize the total community energy costs. The former two systems use game-theoretic approaches where the CES operator (leader) moves first to maximize revenue. Users (followers) then follow the CES operator's actions to independently determine optimal CES energy trading strategies in a non-cooperative finitely repeated game. The fully-competitive CES operator has two degrees of freedom to maximize revenue in a non-cooperative Stackelberg game: energy price and their energy transactions with the grid. Conversely, the benevolent CES operator's ability to maximize revenue is restricted with only one degree of freedom, i.e., their energy transactions with the grid, in a Stackelberg game. The centralized system serves as a baseline to compare the performance of the decentralized game-theoretic energy trading systems. We have the following main contributions in this work:

\begin{enumerate}
\item The non-cooperative Stackelberg game between the fully-competitive CES operator and users has a unique Stackelberg equilibrium where users obtain unique Pareto-optimal Nash equilibrium CES energy trading strategies.
\item Performance analysis demonstrates,
\begin{itemize}
\item Unlike in the centralized cooperative system, both CES operator and users are simultaneously benefited in the fully-competitive system while the grid experiences load leveling.
\item The community economic benefit is greater with the fully-competitive system than the benevolent CES model, and the fully-competitive system can be implemented effectively with least CES battery storage capacity of the three models.
\end{itemize} 
\end{enumerate}

The majority of Stackelberg game-theoretic demand-side management methods in literature exploit demand flexibility of consumers to obtain optimal system-wide objectives \cite{Kilkki, Yu, chen_chen}. For example, the Stackelberg game in \cite{Yu } achieves the optimal load control of electrical appliances through an effective real time pricing method. To the best of our knowledge, few game-theoretic works achieve optimal load management by utilizing energy from distributed energy resources as an alternative to reshaping consumer demand profiles \cite{Atzeni2,Atzeni}. For example, the non-cooperative game in \cite{Atzeni} determines optimal power settings of consumer-owned controllable power sources, such as gas turbines and energy storage devices, to minimize energy costs. In contrast to the use of consumer-owned energy storage devices as in \cite{Atzeni}, we study the use of increased flexibility of a centralized CES device to utilize uncontrollable and intermittent PV power generation to achieve demand-side management without reshaping user demand. To this end, we investigate the leader-follower interaction between the CES operator and the users using Stackelberg games.

This work has two key differences to \cite{TSG1} where a non-cooperative dynamic game between users is studied evaluating only users' autonomy to minimize costs. First, here, we devise bi-level energy trading systems to incorporate autonomies of both CES operator and users to minimize energy costs using Stackelberg games. We also investigate the trade-off between the CES capacity and community benefits, whereas \cite{TSG1} does not impose energy capacity constraints for the CES device by assuming that it has sufficient capacity at all times.

The remainder of this paper is structured as follows. Section~\ref{sec:2} presents related work. The system models are described in Section~\ref{sec:3}, and Section~\ref{sec:4a} describes the centralized energy trading system. The two game-theoretic energy trading systems are discussed in Section~\ref{sec:4}. Section~\ref{sec:5} discusses simulation results, and conclusions are drawn in Section~\ref{sec:6}.

\section{Related Work}\label{sec:2}
There is a rich literature on demand-side management that exploits user demand flexibility to achieve economic power system improvements. For example, dynamic pricing for consumption scheduling \cite{MohsenianRT}, load shifting methods \cite{Babu,Chi,MohensianGT}, and incentive-based demand response programs \cite{Parvania,Caron} have been investigated. We study demand-side management with a CES device to utilize household-distributed PV power generation without modifying users' energy demands.

Prior works have examined centralized control of distributed power resources, such as renewable power sources and storage devices, for effective energy management \cite{Kanchev,Chen}. Decentralized control of energy resources has been proposed to increase system reliability and robustness \cite{smallScale}. In particular, game theory has been applied to analyze interactions between distributed energy resources in power system \cite{gametheoryMicrogrid, Chattopadhyay}. The authors in \cite{Adika} achieve cost-effective energy management through a non-cooperative game that schedules consumer-owned energy storage devices and appliances. In \cite{Rajasekharan}, the authors study a cooperative game-theoretic approach to achieve optimal load balancing using a CES device where users share the stored energy of the CES device to contribute towards community's overall demand-side management. In contrast, we investigate a non-cooperative hierarchical energy trading system between a CES device and users based on Stackelberg game theory. Moreover, compared to \cite{Rajasekharan}, the charging and discharging mechanism of the CES device in this paper employs user-owned PV energy generation to achieve demand-side management.

One branch of demand-side management literature employs Stackelberg game theory to study interactions between utility companies and consumers for optimal demand-side management \cite{Stack3, sabita, GameThPeng,Stack2}. Another branch of Stackelberg game-theoretic demand-side management research explores the interaction of energy consumers with energy market intermediaries such as aggregators that represent demand-side \cite{Stack4, Nekouei, Sarker}. In these papers, the aggregator acts as the middleman between electricity users and the energy market and operates as the energy supplier to the users by selling electricity bought from the utility to downstream consumers. Moreover, an aggregator is responsible to effectively adjust the consumers' aggregated energy demand profile declared to the energy market. To this end, aggregators utilize demand flexibility of users so that the users' compound demand profile satisfies power system needs such as operating limits of the distribution network \cite{Sarker}.

\textit{The proposed energy trading model in this paper has several differences compared to classical aggregation-based approaches.} In particular, the CES operator does not operate as the middleman between the users and the energy market and does not control the amalgamated demand profile of the users. Rather, the CES operator acts as a third-party that encourages users with PV energy generation to trade PV energy with the CES device that can be dispatched to supply peak energy demand of the participating users. \textit{Furthermore, in our system, participating users interact separately with the utility and the CES operator whereas in classical aggregator-based approaches, users merely interact with the aggregator.}

The Stackelberg game between a shared-facility controller and users in \cite{TusharDavid} yields effective demand-side management by managing consumer demand with an energy storage device at the controller-side that is enabled to charge and discharge with the grid. In contrast, in this paper, the charging and discharging mechanism of the CES device is intended to accommodate energy trading strategies from PV energy generation of users. In doing so, we focus on exploiting onsite energy generation from user-owned PV systems for demand-side management as an alternative to energy consumption scheduling of users.

\begin{table*}[t]
\centering
 \caption{table of notation.}
  \begin{tabular}{|m{2 cm}|m{6 cm}||m{2 cm}|m{6 cm}|}
  \hline
Variable/parameter & Definition & Variable/parameter & Definition  \\
    \hline\hline
   
   $\mathcal{A}$ & Set of participating users. & $X_{\mathcal{A}}(t)$ & Aggregate CES energy trading amounts of the users $\mathcal{A}$ at time $t$. \\
   \hline
    
    $\mathcal{P}$ & Set of non-participating users. & $C_{CES}(t)$ & Cost of the CES operator at time $t$. \\
    \hline
    
    $\mathcal{S}(t)$ & Set of participating users with excess PV energy at time $t$. & $C_{\mathcal{A}}(t)$ & Aggregate costs of the users $\mathcal{A}$ at time $t$.\\
    \hline
    
    $\mathcal{D}(t)$ & Set of participating users with energy deficits at time $t$. & $S_{\mathcal{A}}(t)$ & Aggregate surplus energy of the users $\mathcal{A}$ at time $t$.\\
    \hline
    
    $\mathcal{T}$ & Time period of analysis. & $\beta^+,~\beta^-$ & Charging and discharging inefficiencies of the CES device;~ $(0<\beta^+\leq1,~\beta^-\geq1)$. \\
    \hline
    
    $H$ & Total number of time steps in $\mathcal{T}$. &  $\alpha$ & Leakage rate of the CES device;~$(0<\alpha \leq1)$.\\
    \hline
    
    $I$ & Number of users in $\mathcal{A}$.  & $Q_M$ & Maximum energy capacity of the CES device. \\
    \hline
    
    $x_n(t)$ & Energy traded with the CES device by user $n$ at time $t$. & $L_{\text{max}}$ & Maximum allowable grid load.\\
    \hline
    
     $l_n(t)$ & Energy traded with the grid by user $n$ at time $t$. & $v,~w$ & Indices of rows and columns of a matrix.\\
    \hline
    
      $e_n(t)$ & Energy demand of user $n$ at time $t$. & $\delta_t,~\phi_t$ & Positive time-of-use tariff constants of the grid at time $t$.\\
    \hline
    
   $g_n(t)$ & PV energy generation of user $n$ at time $t$. & $\theta$ & Real-valued scalar.\\
    \hline
    
    $s_n(t)$ & Surplus PV energy of user $n$ at time $t$. & $r$ & Iteration number.\\
    \hline
    
      $l_Q(t)$ & Energy traded between the CES device and the grid at time $t$. &  $\tau$ & Small positive value.\\
    \hline
    
     $q(t)$ & CES charge level at the end of time $t$.& $\Gamma$ & Stage game of the repeated game among the users $\mathcal{A}$.\\
    \hline
    
    $L(t)$ & Total grid load at time $t$.&$\mathcal{X}$ & Strategy set of the users $\mathcal{A}$ at time $t$.\\
    \hline
    
       $l_{\mathcal{P}}(t)$ & Total grid load of the users $\mathcal{P}$ at time $t$. & $R$ & Revenue of the CES operator.\\
    \hline
    
 $p(t)$ & Unit grid energy price at time $t$. &$\mathcal{C}$ & Set of cost functions of the users $\mathcal{A}$ at time $t$.  \\
    \hline
    
    $a(t)$ & Unit energy price charged by the CES device at time $t$. &  $\mathcal{Q}$ & Strategy set of the CES operator. \\
    \hline
    
   $C_n(t)$ & Total energy cost of user $n$ at time $t$. & $\Upsilon$ & Stackelberg game between the fully-competitive CES operator and the users $\mathcal{A}$. \\
    \hline
    
    $L_{-n}(t)$ & Total grid load except the grid load of user $n$ at time $t$.  & $\mathcal{L}$ & CES operator. \\
    \hline
    
    $\ve{X_n(t)}$ & Strategy set of user $n$ at time $t$.& $\ve{a}$ & Vector of CES energy price;~$\ve{a}\in \Re^{H\times 1} $. \\
   \hline
   
 $\tilde{x}_n(t)$ & Best response of CES energy trading strategy of user $n$ at time $t$.& $\ve{l_Q}$ & Vector of grid energy of the CES device;~$\ve{l_Q}\in \Re^{H\times 1}$.\\
   \hline
   
    $\ve{x(t)}$ & CES energy trading strategy profile of the users $\mathcal{A}$ at time $t$;~$\ve{x(t)} \in \Re^{1\times I}$. & $\ve{\rho}$ & Matrix of decision variables of the CES operator.\\ 
   \hline 
   \end{tabular}\label{notation}
\end{table*} 
    

\section{System Configuration}\label{sec:3}
In this section, we describe the classification of energy consumers and the models of energy costs and the CES device. The definitions of notations in the subsequent sections are given in Table~\ref{notation}.
\subsection{Demand-Side Model}
The demand-side of the community is divided into participating users $\mathcal{A}$ and non-participating users $\mathcal{P}$. The users $\mathcal{A}$ have their own PV panels without local energy storage devices, and they participate in the energy management optimization by trading energy with the grid and/or the CES device. We assume that each user in $\mathcal{A}$ has a decision-making controller in their household to perform their local energy trading optimization. The users $\mathcal{P}$ consume energy only from the grid as they do not have local power generation capabilities and do not participate in the demand-side management optimization.

The users $\mathcal{A}$ are subdivided into two time-dependent categories: surplus users $\mathcal{S}(t)$ and deficit users $\mathcal{D}(t)$.  We divide the time period $\mathcal{T}$, typically one day, into $H$ equal time steps of length $\Delta{t}$ with discrete time $t = 1,2,\dotsm,H$. We consider $\lvert \mathcal{S}(t) \rvert=I_{\mathcal{S}(t)}$, $\lvert \mathcal{D}(t) \rvert=I_{\mathcal{D}(t)}$, and $\lvert \mathcal{A} \rvert= I = I_{\mathcal{S}(t)}+ I_{\mathcal{D}(t)}$.

At each time $t$, user $i\in \mathcal{S}(t)$ evaluates the optimal energy amount that they can sell to the CES device, and user $j\in \mathcal{D}(t)$ decides optimal energy amount that can be bought from the CES device. These strategies are determined day-ahead, and we assume that the users $\mathcal{A}$ have accurate forecasts of their energy demands and PV power generation for the next day. According to the energy balance at user $n\in \mathcal{A}$
\begin{align}
l_n(t)=x_n(t)+e_n(t)-g_n(t). \label{eq:id2}
\end{align}

Note that $x_n(t)>0$ when the user is charging (or selling energy to) the CES device and $x_n(t)<0$ when discharging (or buying energy from) the CES device.
The surplus energy of $n\in \mathcal{A}$ at time $t$ is given by $s_n(t)=g_n(t)-e_n(t)$.
We specify
\begin{equation}
\begin{split}
0\leq x_i(t) \leq s_i(t), \quad \forall i \in \mathcal{S}(t),~t \in \mathcal{T}, \\
s_j(t)\leq x_j(t) \leq 0, \quad \forall j \in \mathcal{D}(t),~t \in \mathcal{T}. \label{eq:id5}
\end{split}
\end{equation}
\subsection{Community Energy Storage Model}
In this paper, the energy storage model is similar to that in \cite{Atzeni}. At each time $t$, the CES device may exchange energy $l_Q(t)$ with the grid in addition to its energy transactions with the users $\mathcal{A}$. Here, $l_Q(t)>0$ if the CES device is charging from the grid, and $l_Q(t)<0$ if it is selling energy to the grid.

Without loss of generality, consider splitting $x_n(t)$ and $l_Q(t)$ such that $x_n(t)=x_n^+(t)-x_n^-(t)$ and $l_Q(t)=l_Q^+(t)-l_Q^-(t)$ where $x_n^+(t),~l_Q^+(t)\geq 0$ are the charging strategy profiles and $x_n^-(t),~l_Q^-(t)\geq 0$ are the discharging strategy profiles of the CES device at time $t$. Once inefficiencies are introduced, all optimal solutions satisfy $x_n^+(t)x_n^-(t)=0$ and $l_Q^+(t)l_Q^-(t)=0$ at all times to avoid simultaneous charging and discharging of the CES device \cite{Atzeni}. We introduce $\beta^+$ and $\beta^-$ to consider conversion losses of the CES device. For instance, if $x^+$ energy is sold to the CES device, then the charge level only increases by $\beta^+ x^+$. Similarly, $\beta^- x^-$ energy must be discharged to obtain $x^-$ energy from the CES device. If $q(t-1)$ is the charge level at the beginning of time $t$, then $q(t)$ is given by
\begin{equation}
q(t)=\alpha q(t-1)+\beta^+ \chi^+ - \beta^- \chi^- \\
\label{eq:id7}
\end{equation}
where $\chi^+ = \left[\sum_{n=1}^{I}x_n^+(t)+l_Q^+(t)\right]$ and $\chi^- = \left[\sum_{n=1}^{I}x_n^-(t)+l_Q^-(t)\right]$.

Using \eqref{eq:id7}, we write \eqref{eq:id8} to ensure the CES charge level within its energy capacity limit at each time $t$ as
\begin{equation}
\ve{0}\preceq q(0) \ve{\eta}+\Psi \left(\ve{\chi^+}-\ve{\chi^-}\right)\ve{\beta}\preceq \ve{Q_M} \label{eq:id8}
\end{equation}
where $q(0)$ is the initial charge level, and $\ve{Q_M} \in \Re^{H\times 1}$ with all its entries being $Q_M$. Additionally, $\ve{\eta} \in \Re^{H\times 1}$ has $[\ve{\eta}]_v=\alpha^v$, $\Psi\in \Re^{H\times H}$ is a lower triangular matrix that has elements $[\Psi]_{v,w}=\alpha^{v-w}$, $\ve{\beta}=[\beta^+,\beta^-]^T$, $\ve{0}$ is the $H$-dimensional zero column vector, and $\ve{\chi^+},~\ve{\chi^-}$ are the $H$-dimensional column vectors with all their entries being $\chi^+,~\chi^-$, respectively. Note that \eqref{eq:id8} is the vector algebra form of the linear recurrence relation \eqref{eq:id7}.

Assuming $q(0)$ is within the CES device's safe operating region, we set \eqref{eq:id9} to ensure the continuous operation of the CES device for the next day and to prevent over-charging or over-discharging during $\mathcal{T}$ \cite{Chen}
\begin{equation}
q(H)=q(0). \label{eq:id9}
\end{equation}

\subsection{Energy Cost Models}\label{sec:3a}

The unit electricity price of the grid at time $t$ is assumed to have a constant baseline component and a variable real-time component that is proportional to the total grid load at time $t$ \cite{MohensianGT,Atzeni}. In this work, the total grid load at time $t$ is $L(t)=\sum_{n=1}^I l_n(t)+l_Q(t)+l_{\mathcal{P}}(t)$. In this paper, we assume, at each time $t$, $0<L(t)$ for non-negative grid pricing and $L(t)<L_{\text{max}}$ where $L_{\text{max}}$ is the maximum allowable load on the grid without compromising voltage and line capacity limits of the grid. The unit electricity price of the grid at time $t$ is given by $p(t)=\phi_t L(t)+\delta_t$ where $\phi_t$ and $\delta_t$ are determined according to a day-ahead market clearing process \cite{Atzeni}. With similar analysis to \cite{MohensianGT}, the resulting grid energy cost function at time $t$, $p(t)L(t)$, is a strictly convex function with respect to $L(t)$.

In our game-theoretic systems, the CES operator adopts prices for energy transactions with the users $\mathcal{A}$. Then, in these systems $C_n(t)$ is given by
\begin{equation}
C_n(t)=p(t)l_n(t)-a(t)x_n(t) \label{eq:id11}
\end{equation}
where $l_n(t)$ is given by \eqref{eq:id2}. Both fully-competitive and benevolent CES operators obtain revenue through energy trading with the grid and the users $\mathcal{A}$. Assuming that the CES operator exchanges energy with the grid at the grid energy price, we consider the CES revenue as
\begin{equation}
R=\sum_{t=1}^H {\left(-a(t)\sum_{n=1}^I x_n(t)-p(t)l_Q(t)\right)}. \label{eq:id12}
\end{equation}

However, in the centralized energy trading system, the CES operator does not obtain revenue from energy trading with the users $\mathcal{A}$. In this regard, we do not consider a separate revenue function as \eqref{eq:id12}, and the cost of user $n\in \mathcal{A}$ is derived as in \eqref{eq:id11} disregarding the term $a(t)x_n(t)$.

\section{Centralized energy trading system}\label{sec:4a}

This section describes the community energy trading system with a centralized cooperative CES operator that solves the optimization problem of minimizing the total energy cost paid by the entire community to the grid. Note that this centralized approach serves as a baseline to compare the performance of the decentralized game-theoretic systems in Section~\ref{sec:4}. Here, we assume that the users $\mathcal{A}$ communicate their energy demand and PV energy generation profiles to the CES operator, and the operator also has the perfect knowledge of the community participation percentage. The CES operator schedules the energy transactions across the community by solving the optimization problem
\begin{equation}
\min \sum_{t=1}^H{p(t)L(t)} \label{eq:id12a}
\end{equation}
subjects to constraints \eqref{eq:id5}, \eqref{eq:id8}, and \eqref{eq:id9}.

Note that in this system, \eqref{eq:id12a} does not include a price signal for the CES operator's energy transactions with the users $\mathcal{A}$ and consequently, the operator has no direct incentive. It also requires impractical information exchange and cooperation. The cooperating participating users similarly do not have direct incentives as their personal energy costs may inflate for the benefit of the overall community. All of these reasons make the centralized cooperative energy trading system less feasible. Even though this centralized approach is not appropriate for the general circumstances of the considered energy trading scenario, it is still a potential implementation of the energy trading between the CES device and the users $\mathcal{A}$.

\section{Decentralized energy trading systems}\label{sec:4}

In our decentralized energy trading systems, the CES operator, as the leader, interacts with the users $\mathcal{A}$ to maximize their revenue \eqref{eq:id12}. The users $\mathcal{A}$ follow the leader's actions to minimize their individual energy costs in \eqref{eq:id11} by manipulating $x_n(t)$. We develop Stackelberg game-theoretic frameworks to analyze the hierarchical CES-user energy trading interactions. To derive the solutions to the Stackelberg games, we adopt insights from backward induction \cite{gametheoryessentials}. To this end, first, the actions of the users $\mathcal{A}$ are derived based on the knowledge of actions of the CES operator. Then our analysis proceeds backward to determine the actions of the CES operator.

\subsection{Objective of the Participating Users}

Here, each user $n\in \mathcal{A}$ seeks to minimize their personal energy costs. Therefore, in response to any suitable $\ve{\rho}=[\ve{a},~\ve{l_Q}]$ of the CES operator, user $n\in \mathcal{A}$ minimizes their energy cost in \eqref{eq:id11} at each time $t$. The cost function \eqref{eq:id11} is quadratic with respect to both $l_n(t)$ and $x_n(t)$. We consider
\begin{equation}
C_n(t)=K_2 l_n(t)^2 + K_1 l_n(t)+K_0 \label{eq:id13}
\end{equation}
where $K_2=\phi_t$, $K_1=(\phi_tL_{-n}(t)+\delta_t-a(t))$, and $K_0=-a(t)s_n(t)$.

Since \eqref{eq:id13} depends on the actions of the other users $n'\in \mathcal{A}\backslash n$, we formulate a non-cooperative game $\Gamma\equiv\langle \mathcal{A},\mathcal{X},\mathcal{C}\rangle$ among the users $\mathcal{A}$ at each time $t\in \mathcal{T}$ to determine their optimal strategies. Here, $\mathcal{X}=\prod_{n=1}^I \ve{X_n(t)}$ where $\ve{X_n(t)}$ of user  $n\in \mathcal{A}$ subject to constraints \eqref{eq:id5}, and $\mathcal{C}=(C_1(t),\ldots,C_I(t))$. We denote $\ve{x(t)}=[x_1(t),\dotsm,x_I(t)]$. Each user $n\in \mathcal{A}$ selects their strategy $x_n(t)\in\ve{X_n(t)}$ to minimize the cost function $C_n(x_n(t),\ve{x_{-n}(t)})\equiv C_n(t)$. Here, $\ve{x_{-n}(t)}$ is the CES energy transaction strategy profile of the users $n'\in \mathcal{A}\backslash n$. Therefore, each user $n\in \mathcal{A}$ determines
\begin{equation}
\tilde{x}_n(t)=\argmin_{x_n(t)\in \ve{X_n(t)}}C_n(x_n(t),\ve{x_{-n}(t)}).\label{eq:id13b}
\end{equation}

To make the game-theoretic analysis tractable, we assume that the users $\mathcal{A}$ have accurate day-ahead predictions of PV power generation and energy demand. Consequently, playing the game $\Gamma$ at each time $t = 1,2,\dotsm,H$ by the users $\mathcal{A}$ using $\ve{\rho}$ turns into a non-cooperative finitely repeated game with perfect information where $\Gamma$ is the stage game \cite{gametheoryessentials}.

\begin{proposition}
For any given values of $a(t)$ and $l_Q(t)$, the stage game $\Gamma$ obtains a unique pure-strategy Nash equilibrium.
\end{proposition}

\begin{proof} Nash equilibrium implies no player can gain by unilaterally changing their own strategy while the others play their Nash equilibrium strategies \cite{gametheoryessentials}. For feasible $\ve{x_{-n}(t)}$, \eqref{eq:id13} is strictly convex since its second derivative with respect to $x_n(t)$ is positive as $\phi_t>0$ \cite{boyd2004convex}. Therefore, each participating user's objective function in \eqref{eq:id13b} is strictly convex. Additionally, the individual strategy sets are compact and convex due to linear inequalities \eqref{eq:id5}. Therefore, a unique Nash equilibrium with pure strategies for the game $\Gamma$ is obtained \cite{rosen}.
\end{proof}

The best response $\tilde{x}_n(t)$ of user $n\in \mathcal{A}$ to $\ve{x_{-n}(t)}$ can be found using
\begin{equation}
\frac{\partial C_n(t)}{\partial x_n(t)}\bigg|_{x_n(t)=\tilde{x}_n(t)}=2K_2 (\tilde{x}_n(t)-s_n(t))+K_1=0. \label{eq:id15}
\end{equation}
By solving \eqref{eq:id15} for all participating users $I$, using the expressions of $K_1$ and $K_2$ in \eqref{eq:id13}, the optimal response of user $n\in \mathcal{A}$ at the Nash equilibrium, $\tilde{x}^{*}_n(t)$, can be written as a function of the CES operator's output variables $a(t)$ and $l_Q(t)$
\begin{gather}
\tilde{x}^{*}_n(t)=s_n(t)-\varepsilon(t),\nonumber \\
\varepsilon(t)=-(I+1)^{-1}[\phi_t^{-1}(a(t)-\delta_t)-l_{\mathcal{P}}(t)-l_Q(t)]. \label{eq:id17}
\end{gather}

Given the Nash equilibrium as in \eqref{eq:id17}, parameters $\delta_t$ and $\phi_t$ can be chosen such that $\tilde{x}^*_n(t)$ satisfies \eqref{eq:id5} for given $a(t),~l_Q(t)$, and $l_{\mathcal{P}}(t)$. However, for general application of our system, we consider constraints in the CES operator's revenue maximization problem so that the CES operator's selection of $a(t)$ and $l_Q(t)$ assures that the Nash equilibrium \eqref{eq:id17} satisfies \eqref{eq:id5}. This procedure is explained in the next subsection.

\subsection{Objective of the Community Energy Storage Operator}\label{CES}

In the decentralized energy trading setting, the CES operator's primary objective is to maximize the revenue in \eqref{eq:id12}. According to backward induction, if we substitute \eqref{eq:id17} into \eqref{eq:id12}, the CES operator's utility maximization can be simplified to a quadratic optimization problem to determine
\begin{equation}
\ve{\rho^*}=\argmax_{\ve{\rho}\in\mathcal{Q}}{\sum_{t=1}^H (\lambda a(t)^2+\mu a(t)+\nu l_Q(t)^2+\xi l_Q(t))} \label{eq:id22}
\end{equation}
where $\lambda=-I(I+1)^{-1}\phi_t^{-1}$, $\mu=I(I+1)^{-1}(l_\mathcal{P}(t)+\phi_t^{-1}\delta_t)-\sum_{n=1}^I{s_n(t)}$, $\nu=-\phi_t(I+1)^{-1}$, and $\xi =-(I+1)^{-1}(\phi_t l_\mathcal{P}(t)+\delta_t)$. To ensure that users' actions obtained from \eqref{eq:id17} satisfy \eqref{eq:id5}, at each time $t$, the CES operator selects $a(t)$ and $l_Q(t)$ such that 

\begin{equation}
\begin{rcases}
               \textrm{max}[\{s_j(t)\}_I]\leq\varepsilon(t)\leq0,~~\text{if}~\mathcal{A}=\mathcal{D}(t), \\
                0\leq\varepsilon(t)\leq \textrm{min}[\{s_i(t)\}_I],~~~\text{if}~\mathcal{A}=\mathcal{S}(t),\\
                \varepsilon(t)=0,~~~~~~~~~~~~~~~~~~~~~~~\text{otherwise.}
\end{rcases}\label{eq:id22a}
\end{equation}
Hence, in addition to \eqref{eq:id8} and \eqref{eq:id9}, we consider \eqref{eq:id22a} in $\mathcal{Q}$ depending on the nature of the users $\mathcal{A}$ at time $t$.

The objective function in \eqref{eq:id22} is strictly concave since its Hessian matrix is negative definite for all $\ve{a},\ve{l_Q}\in\mathcal{Q}$ as coefficients $\lambda,\nu<0$. Moreover, $\mathcal{Q}$ is non-empty, closed and convex as it is only subject to linear constraints. Therefore, \eqref{eq:id22} always has a unique maximum \cite{boyd2004convex}. It is important to note that \eqref{eq:id22} is a quadratic program and hence, it can be solved by using non-linear programming methods such as the interior point method described in \cite{boyd2004convex}.

\subsection{Benevolent CES Operator Model}

After describing the objectives of the users $\mathcal{A}$ and the CES operator, in this section, we analyze the Stackelberg energy trading competition between the benevolent CES operator and the users $\mathcal{A}$. Consider if $x_n(t)=s_n(t),~\forall n\in \mathcal{A}$ and $\forall t\in \mathcal{T}$, then these user strategies are at the Nash equilibrium of the game $\Gamma$ if and only if $\varepsilon(t)=0$ at each time $t$. As a result
\begin{equation}
a(t)=\delta_t+\phi_t (l_Q(t)+l_\mathcal{P}(t)). \label{eq:id21}
\end{equation}

In this model, the price constraint \eqref{eq:id21} applies as an auxiliary constraint for the CES operator at each time $t$ when maximizing their revenue in addition to \eqref{eq:id8} and \eqref{eq:id9}. Then, the objective function in \eqref{eq:id22} can be written as a function of $l_Q(t)$ only by substituting \eqref{eq:id21} into \eqref{eq:id22}. Thus, the objective of the CES operator in the benevolent scenario is to find
\begin{equation}
\ve{\bar{l}_Q}=\argmax_{\ve{l_Q}\in\mathcal{Q}}{\sum_{t=1}^H (\gamma_1 l_Q(t)^2+\gamma_2 l_Q(t)+\gamma_3)} \label{eq:id21a}
\end{equation}
where $\gamma_1=-\phi_t,~\gamma_2=-\delta_t-\phi_t(l_\mathcal{P}(t)+\sum_{n=1}^I{s_n(t)})$, and $\gamma_3=-\sum_{n=1}^I{s_n(t)}(\delta_t+\phi_tl_\mathcal{P}(t))$.

In the Stackelberg game, the CES operator, as the leader, firstly determines optimal $l_Q(t)$ by solving \eqref{eq:id21a} and then their energy price $a(t)$ using \eqref{eq:id21} for each time $t\in \mathcal{T}$. Using these values, the users $\mathcal{A}$ determine optimal strategies of $x_n(t)$ by playing the game $\Gamma$ at each time $t$.  As a result, at the Nash equilibrium of the game $\Gamma$, the energy requirements of the users $\mathcal{A}$ are shifted on to the CES device such that $\tilde{x}^{*}_n(t)=s_n(t)$. Hence, $\tilde{l}^*_n(t)=(\tilde{x}^{*}_n(t)-s_n(t))=0$,~$\forall n \in \mathcal{A},~\forall t \in \mathcal{T}$.

In this system, the CES operator does not have full freedom to maximize the revenue in \eqref{eq:id22} with the additional constraint \eqref{eq:id21}. Here, $\bar{p}(t)=\delta_t+\phi_t(l_\mathcal{P}(t)+\bar{l}_Q(t))=\bar{a}(t)$ where $\bar{l}_Q(t)$ is the CES device's grid load at time $t$ at its maximum revenue in \eqref{eq:id21a} since the users $\mathcal{P}$ and the CES device are the only remaining energy loads on the grid as the users $\mathcal{A}$ shift their net energy requirements $s_n(t)$ to the CES device. We consider this as a benevolent CES operator that is regulated by the users $\mathcal{A}$ to amalgamate all their energy requirements into one entity with storage capabilities for better demand-side management.

\subsection{Fully-competitive CES Operator Model}\label{SecC}

This section explains the non-cooperative Stackelberg game between the fully-competitive CES operator and the users $\mathcal{A}$. In the system, the CES operator first sets $\ve{\rho}$ to maximize their revenue in \eqref{eq:id22} and broadcasts them to the users $\mathcal{A}$. Using these signals, the users $\mathcal{A}$ repeat the non-cooperative game $\Gamma$ at each time $t$. In contrast to the benevolent CES model, in this system, the CES operator regulates $\varepsilon(t)$ by considering the nature of the users $\mathcal{A}$ at each time $t$ as given in \eqref{eq:id22a}. Hence, the CES operator's objective is identical to \eqref{eq:id22}.

In this scenario, the bi-level interaction between the CES operator and the users $\mathcal{A}$ can be formulated as a non-cooperative Stackelberg game $\Upsilon$. We represent the strategic form of $\Upsilon$ as $\Upsilon\equiv\langle \{\mathcal{L},\mathcal{A}\},\{\mathcal{Q},\mathcal{X}\},\{R,\mathcal{C}\}\rangle$ where the CES operator $\mathcal{L}$ is the leader, the users $\mathcal{A}$ are the followers, and all other notations are defined as in the preceding sections.
\begin{proposition}
The game $\Upsilon$ obtains a unique Stackelberg equilibrium.
\end{proposition}

\begin{proof} The non-cooperative game $\Gamma$ between the users $\mathcal{A}$ has a unique Nash equilibrium for given $a(t)$ and $l_Q(t)$ of the CES operator (see Proposition~1). There is also a unique solution of the CES operator's revenue maximization \eqref{eq:id22}. Therefore, the game $\Upsilon$ obtains a unique Stackelberg equilibrium as soon as the CES operator determines their unique revenue maximizing strategy $\ve{\rho^*}$ while the users $\mathcal{A}$ play their unique Nash equilibrium strategy profile $\ve{x^*(t)}$ at each time $t$.
\end{proof}

Note that, according to backward induction, the Stackelberg equilibrium strategies of the CES operator and the users $\mathcal{A}$, $\ve{\rho^*}$ and $x_n^*(t)\in \ve{x^*(t)}$, are equal to the solution of \eqref{eq:id22} and the resulting $\tilde{x}^{*}_n(t)$ after substituting $a^*(t),~l_Q^*(t) \in \ve{\rho^*}$ in \eqref{eq:id17}, respectively. The Stackelberg equilibrium strategies satisfy

\begin{gather}
C_n(\ve{x^*(t)},~\ve{\rho^*}) \leq C_n(x_n(t),\ve{x_{-n}^*(t)}, \ve{\rho^*}), \nonumber \\
\forall n\in \mathcal{A},~\forall x_n(t)\in \ve{X_n(t)},~\forall t\in \mathcal{T},\label{eq:id23}
\end{gather}
 \begin{equation}
 \begin{split}
R(\ve{X^*},\ve{\rho^*}) \geq R(\ve{X^*},\ve{\rho}),~\forall \ve{\rho} \in \mathcal{Q}
\end{split}\label{eq:id23b}
\end{equation}
where $\ve{x_{-n}^*(t)}$ is the Nash equilibrium strategy profile of the users $\mathcal{A}$ except user $n$ at time $t$, and $\ve{X^*}=(\ve{x^*(1)}^T,\ldots,\ve{x^*(H)}^T)$ is the $H$-tuple of Nash equilibrium strategy profiles of the users $\mathcal{A}$ at each $t$ in response to $\ve{\rho^*}$.
\begin{proposition}
At the Stackelberg equilibrium of the game $\Upsilon$, the Nash equilibrium CES energy trading strategy profile $\ve{x^*(t)}$ of the users $\mathcal{A}$ achieved by playing the game $\Gamma$ is Pareto optimal.
\end{proposition}

\begin{proof} Let us consider the energy trading strategies of the users $\mathcal{A}$ and the CES operator at the Stackelberg equilibrium of the game $\Upsilon$ at any $t\in \mathcal{T}$: $[\ve{x^*(t)},a^*(t),l_Q^*(t)]$. Assume there is any feasible $\ve{x^{'}(t)} (\neq\ve{x^*(t)})$ such that $\ve{x^{'}(t)}$ Pareto dominates $\ve{x^*(t)}$. Define $X_{\mathcal{A}}(t)$ (see Table~\ref{notation}) at $\ve{x^{'}(t)}$ and $\ve{x^*(t)}$ as $X^{'}_{\mathcal{A}}(t)$ and $X^*_{\mathcal{A}}(t)$, respectively, where $X^{'}_{\mathcal{A}}(t)= X^*_{\mathcal{A}}(t)+\theta X^*_{\mathcal{A}}(t)$. First, assume $\theta$ is a non-zero scalar. Due to the introduced change of $X^*_{\mathcal{A}}(t)$ to $X^{'}_{\mathcal{A}}(t)$, the CES device is forced to over-charge or over-discharge violating \eqref{eq:id9}. Therefore, the CES operator has to divert from their strategy $l_Q^*(t)$ to $l_Q^{'}(t)=l_Q^*(t)-\theta X^*_{\mathcal{A}}(t)$ in order to satisfy \eqref{eq:id9} that ensures the sustainable operation of the CES device throughout the day. For example, if the CES charge level rises due to the introduced change, the CES operator has to discharge the increased amount of energy to the grid. Note that, in doing so, the Stackelberg equilibrium grid price $p^*(t)$ does not change as the total grid load does not change.

Then at the new operating point $C_{CES}(t)$ is given by
 \begin{gather}
C_{CES}(t)=a^*(t)X^{'}_{\mathcal{A}}(t)+p^*(t)l_Q^{'}(t),\label{eq:id23c}
\end{gather}
and $C_{\mathcal{A}}(t)$ is
\begin{gather}
C_{\mathcal{A}}(t)=p^*(t)\left(X^{'}_{\mathcal{A}}(t)-S_{\mathcal{A}}(t)\right)-a^*(t)X^{'}_{\mathcal{A}}(t)\label{eq:id23d}
\end{gather}
where $S_{\mathcal{A}}(t)={\sum_{n=1}^Is_n(t)}$. By substituting $X^{'}_{\mathcal{A}}(t)=X^*_{\mathcal{A}}(t)+\theta X^*_{\mathcal{A}}(t)$ and $l_Q^{'}(t)=l_Q^*(t)-\theta X^*_{\mathcal{A}}(t)$ into \eqref{eq:id23c} and \eqref{eq:id23d}, it is evident that if $C_{\mathcal{A}}(t)$ decreases, then $C_{CES}(t)$ increases and vice versa. Such a situation is not led by the CES operator and hence there is no feasible $\ve{x^{'}(t)}\in \mathcal{X}\backslash\ve{x^*(t)}$ that Pareto dominates $\ve{x^*(t)}$ after obtaining the Stackelberg equilibrium. 

Moreover, when $\theta =0$, $X^{'}_{\mathcal{A}}(t)= X^*_{\mathcal{A}}(t)$. This implies that if at least one user in $\mathcal{A}$ changes their Nash equilibrium strategy to another strategy, then the other users have to change their aggregate CES energy trading strategies by the same energy amount such that $X^{'}_{\mathcal{A}}(t)= X^*_{\mathcal{A}}(t)$. In this situation, the grid price $p^*(t)$ and the unit energy price of the CES device $a^*(t)$ are unchanged. Hence, a reduction in one user's energy cost due to a change in their operating point would lead to an increase in the cost of at least one of the other users. Hence, with $\theta=0$, for given $[a^*(t),l_Q^*(t)]$, it is infeasible to adopt any $\ve{x^{'}(t)}\in \mathcal{X}\backslash\ve{x^*(t)}$ such that $C_n(\ve{x^{'}(t)})\leq C_n(\ve{x^*(t)}), \forall n\in \mathcal{A}$ and $C_n(\ve{x^{'}(t)})< C_n(\ve{x^*(t)})$ for some $n\in \mathcal{A}$. This concludes the proof of the proposition.
\end{proof}

A two-step iterative algorithm was used to determine the Stackelberg equilibrium in the game $\Upsilon$ as shown in Algorithm~\ref{Algo1}. In the first step, the CES operator determines $\ve{\rho}$ that maximizes \eqref{eq:id12} under constraints \eqref{eq:id8}, \eqref{eq:id9}, and \eqref{eq:id22a} by using the CES energy trading strategies of the users $\mathcal{A}$. In the second step, the users $\mathcal{A}$ solve \eqref{eq:id13b} using \eqref{eq:id17}. To achieve the convergence of the algorithm, similar to \cite{Atzeni}, we adopt the termination criterion ${\|\ve{\rho}^{(r)}-\ve{\rho}^{(r-1)}\|}_2/{\|\ve{\rho}^{(r)}\|}_2 \leq\tau$ where $\ve{\rho}^{(r)}$ represents $\ve{\rho}$ calculated at the iteration $r$.

The users' response in \eqref{eq:id17} is their optimal response at each time $t$ for a given feasible $\ve{\rho}$ by the CES operator. This implies that, in Algorithm~\ref{Algo1}, the CES energy trading strategies of the users $\mathcal{A}$ will converge once $\ve{\rho}$ converges to a fixed point. Consequently, Algorithm~\ref{Algo1} converges as the variation of $\ve{\rho}$ decreases after a certain number of iterations. Hence, Algorithm~\ref{Algo1} converges approximately to a fixed point once the termination criterion is satisfied for sufficiently small $\tau$.

\begin{algorithm}
\caption{Game to obtain the Stackelberg equilibrium}\label{Algo1}
\begin{algorithmic}[1]
\renewcommand{\algorithmicrequire}{\textbf{Step 1:}}
\renewcommand{\algorithmicensure}{\textbf{Step 2:}}
\REQUIRE
\STATE $r \leftarrow  1$.
\NoThen
\IF{$r=1$}
\STATE{The CES operator selects a feasible starting point for $\ve{\rho}$ and broadcasts it to the users $\mathcal{A}$.}
\ELSE
\STATE The CES operator maximizes \eqref{eq:id12} subject to constraints \eqref{eq:id8}, \eqref{eq:id9}, and \eqref{eq:id22a} using $\tilde{\ve{X}}^*$ and broadcasts  $\ve{\rho}$ to the users $\mathcal{A}$.
\ENDIF
\ENSURE
\STATE Each user $n\in \mathcal{A}$ determines $\tilde{x}^{*}_n(t)$ at each time $t$ using \eqref{eq:id17} and $\ve{\rho}$, and $\tilde{\ve{X}}^*=\big(\ve{\tilde{x}^{*}(1)}^T,\dotsm,\ve{\tilde{x}^{*}(H)}^T\big)$ is announced to the CES operator.
\STATE $r \leftarrow r+1$. 
\STATE Repeat from 2 until  ${\|\ve{\rho}^{(r)}-\ve{\rho}^{(r-1)}\|}_2/{\|\ve{\rho}^{(r)}\|}_2 \leq\tau$.
\STATE Return $\tilde{\ve{X}}^*$ and $\ve{\rho}$ as the Stackelberg equilibrium.
\end{algorithmic}
\end{algorithm}

\section{Results and discussion}\label{sec:5}

For numerical simulations, we consider a residential community of 40 people with 30\%, 40\%, and 50\% participating users. We obtain the average daily domestic PV power generation and user electricity demand profiles from \cite{WPNdata}. We use $H=48$, $\Delta t=30~\textrm{min}$, $Q_M=80$~kWh, $q(0)=0.25Q_M$, $\alpha=0.9^{1/48}$, $\beta^+=0.9$, $\beta^-=1.1$ \cite{Atzeni}, and $\tau=0.002$. Parameter $\phi_t$ is selected in each case such that $\phi_{\text{peak}}=1.5\times \phi_{\text{off peak}}$ where the peak period is 16:00-23:00. The value of $\phi_{\text{peak}}$ is then set such that the predicted daily unit price range of the grid is the same as a reference time-of-use unit electricity price range used in Sydney, Australia \cite{ausgrid}. $\delta_t$ is a constant across time such that predicted average grid price matches the average price of the reference signal. To compare results, we consider a baseline energy trading system without a CES device where the users $\mathcal{A}$ trade energy exclusively with the main power grid that has the same energy cost model in Section~\ref{sec:3}. In particular, PV energy producers sell all surplus PV energy directly to the grid.

\subsection{Performance of Algorithm~\ref{Algo1}}
\begin{figure}[t!]
\centering
\includegraphics[width=0.8\columnwidth]{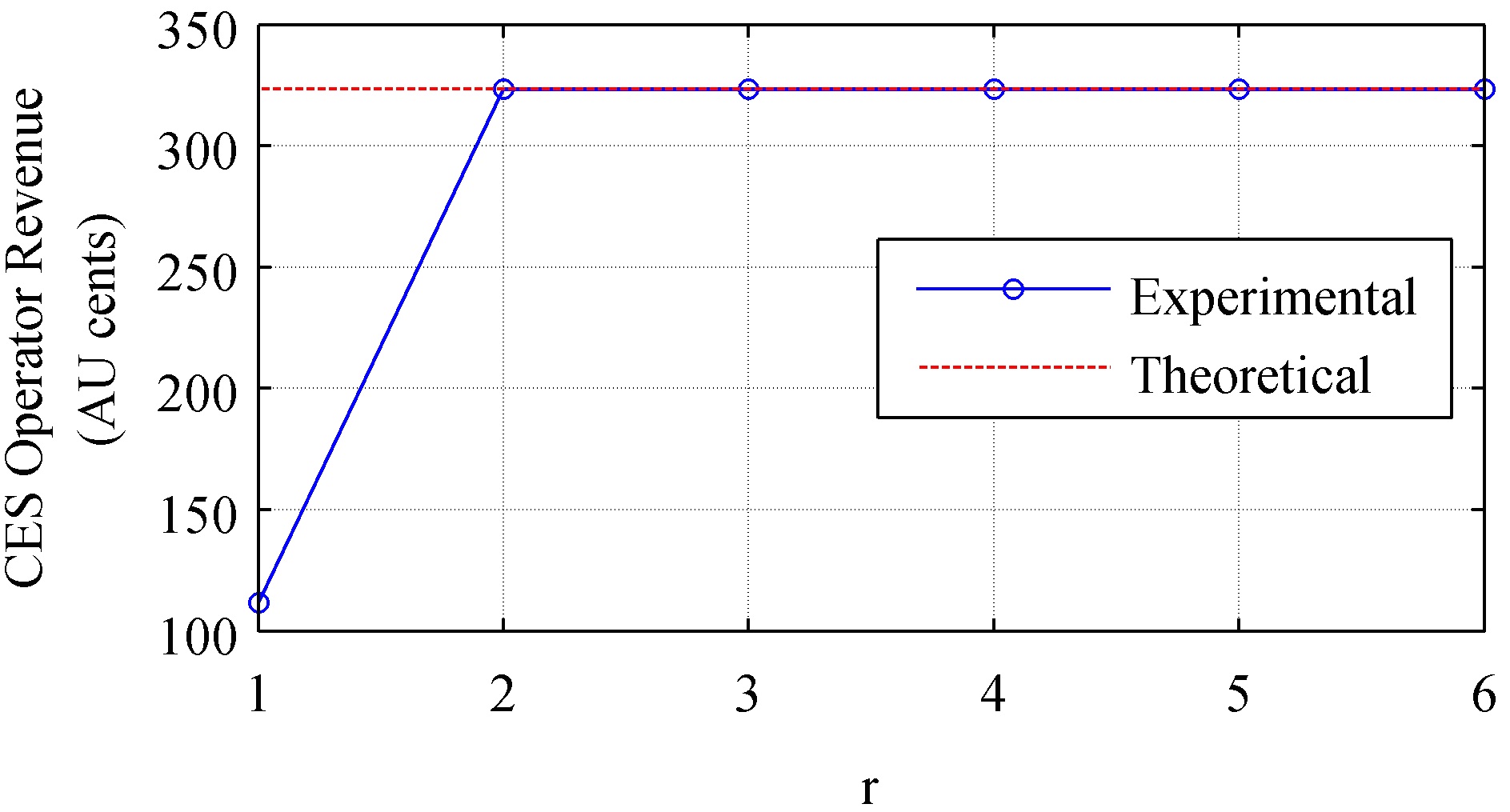}
\caption{Convergence of Algorithm 1 with 30\% participating users in the fully-competitive CES model.}
\label{fig:Algo1}
\end{figure}
In this section, the convergence of Algorithm~\ref{Algo1} is examined. To show that Algorithm~\ref{Algo1} converges to the theoretical optimum at the Stackelberg equilibrium of the fully-competitive CES model, we consider the optimization \eqref{eq:id22} that gives the CES operator's theoretical optimal revenue at the Stackelberg equilibrium according to backward induction. To solve the optimization problem given in \eqref{eq:id22}, we used the interior-point algorithm specified in the ``fmincon" solver in the MATLAB optimization toolbox. Fig.~\ref{fig:Algo1} illustrates that Algorithm~\ref{Algo1} reaches the theoretical optimal revenue of the CES operator within the first 2 iterations with 30\% participating users. Here, as the initial conditions, the vector $\ve{l}^{(1)}_{\ve{Q}}$ was considered as $\ve{0}$. Then, the elements of $\ve{a}^{(1)}$ were selected such that $\tilde{x}^{*}_n(t)$ values, obtained by using $\ve{l}^{(1)}_{\ve{Q}}$,~$\ve{a}^{(1)}$, and \eqref{eq:id17}, satisfy \eqref{eq:id8},~\eqref{eq:id9}, and \eqref{eq:id22a}. 
However, Algorithm~\ref{Algo1} reaches the termination criterion ${\|\ve{\rho}^{(r)}-\ve{\rho}^{(r-1)}\|}_2/{\|\ve{\rho}^{(r)}\|}_2 \leq\tau$ after $r=4$ iterations because the CES operator adjusts their strategy $\ve{\rho}$ until the termination criterion is achieved. Similarly, the algorithm converges after $r=11$ and $r=13$ iterations when the community has 40\% and 50\% participating users, respectively.

\subsection{Preliminary Study of Three Energy Trading Systems}
Consider the case with 40\% participating users to demonstrate our CES models. In Fig.~\ref{fig:40PUprices}, we illustrate the variations of price signals of the CES operator and the grid in the fully-competitive CES model and the grid price of the baseline. 
\begin{figure}[b!]
\centering
\includegraphics[width=0.8\columnwidth]{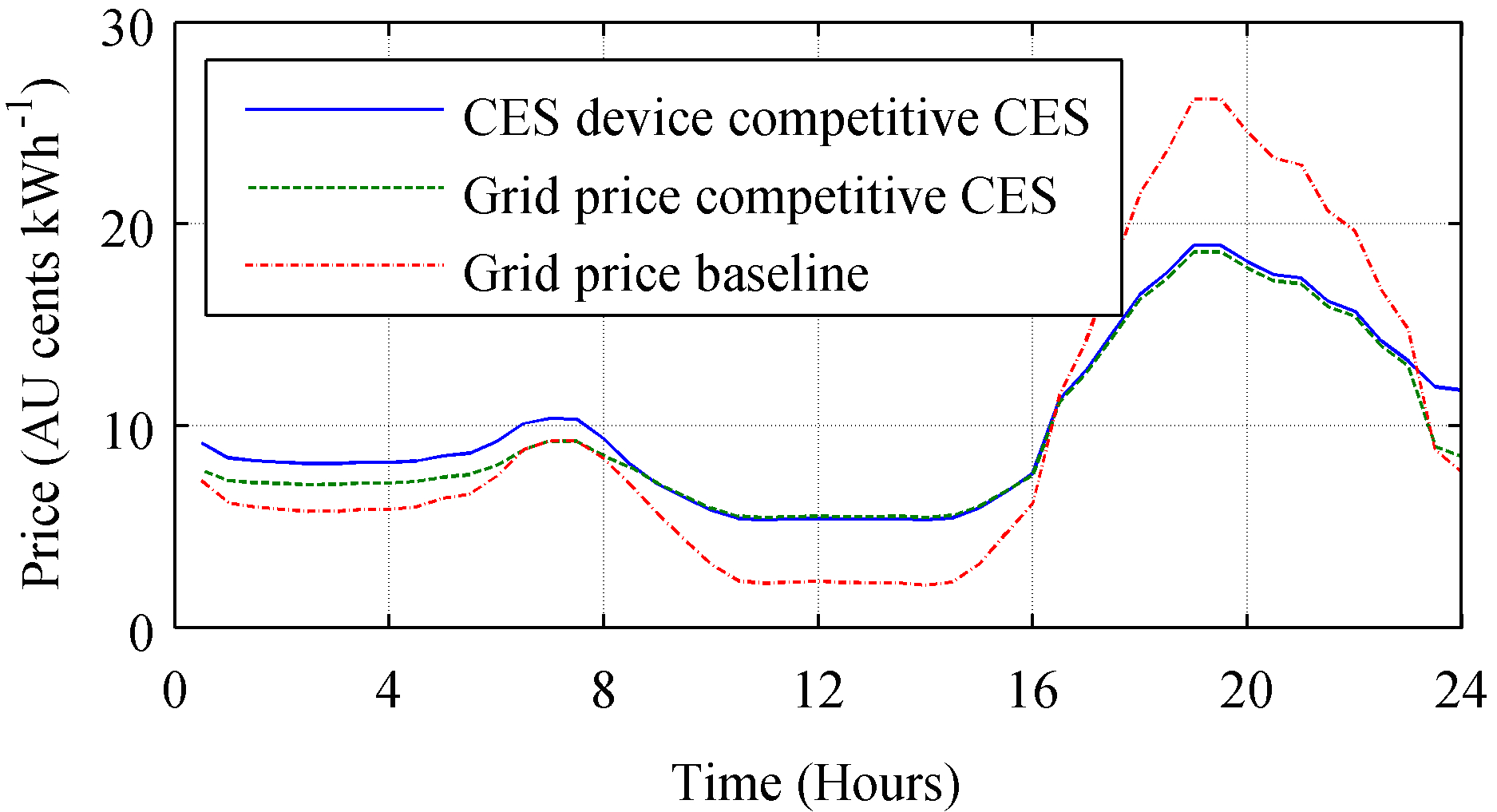}
\caption{Variation of electricity prices with 40\% participating users in the fully-competitive CES model and the baseline.}
\label{fig:40PUprices}
\end{figure}
Fig.~\ref{fig:40PUprices} shows that the introduction of the CES device reduces the peak grid electricity price in the competitive CES model compared to the baseline. Before 09:00, when there is little PV energy and all participating users are deficit users, the CES operator sets a price above the equilibrium grid price such that it is unfavorable for any of the deficit users to purchase energy. Subsequently, the users $\mathcal{A}$ may not buy energy from the CES device during this period. During the day, when PV energy is plentiful, and through the evening peak, when electricity demand is greatest, it is favorable for the CES operator to trade energy with the users $\mathcal{A}$. Therefore, at these times the CES price approaches the equilibrium grid price in a similar way to the benevolent CES model. In turn, the users $\mathcal{A}$ transfer the majority of their energy transactions to the CES device. In the benevolent and the centralized cooperative CES models, there are more energy transactions between the CES device and the users $\mathcal{A}$ than in the fully-competitive CES model. Fig.~\ref{fig:40PUgridPrice} shows that this causes a greater reduction in the peak grid price.
\begin{figure}[b!]
\centering
\includegraphics[width=0.8\columnwidth]{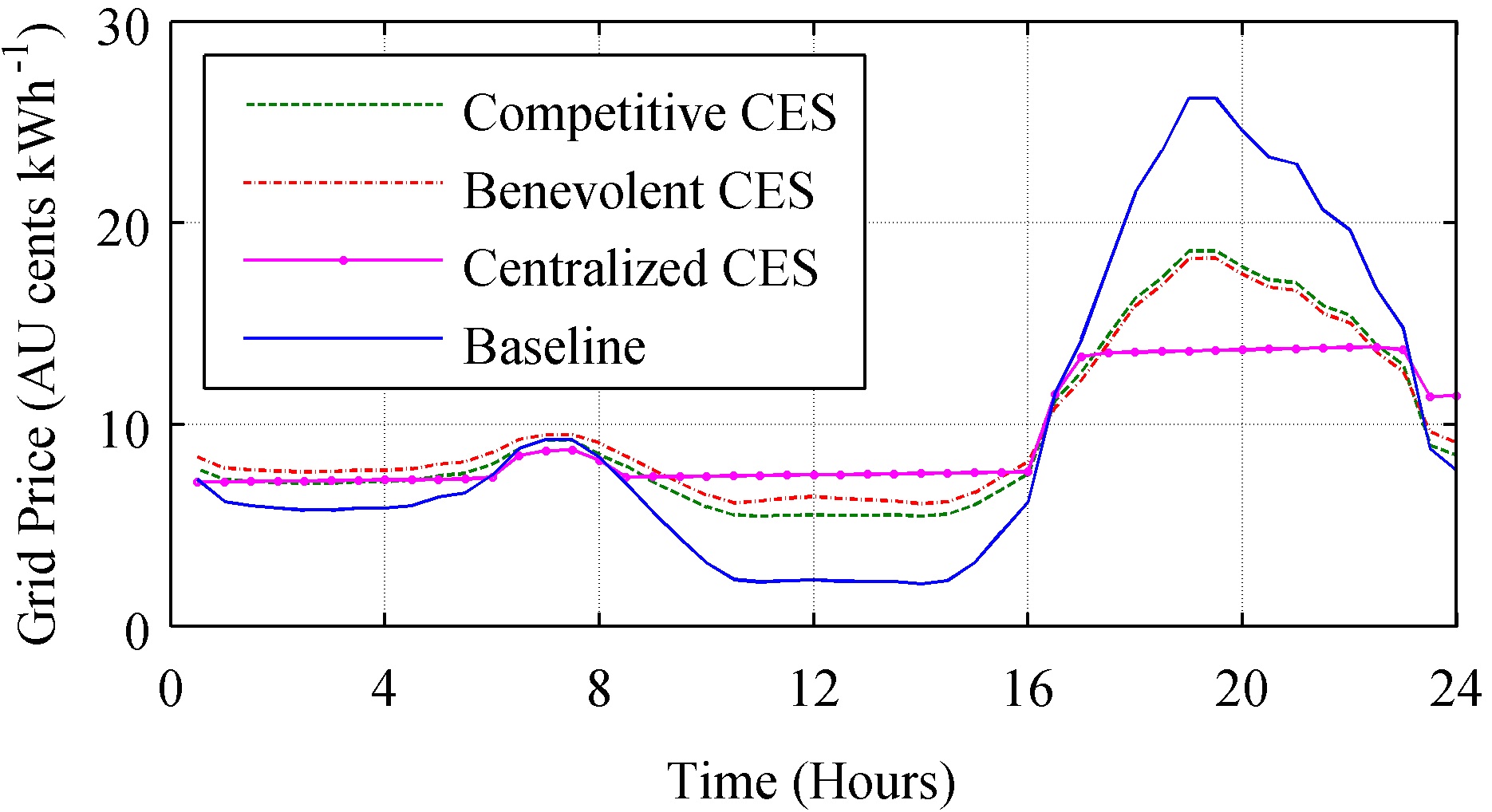}
\caption{Variation of grid electricity prices of the different CES operator models with 40\% participating users.}
\label{fig:40PUgridPrice}
\end{figure}

More energy transactions between the CES device and the users $\mathcal{A}$ require greater CES storage capabilities (see Fig.~\ref{fig:40PUCEScharge}), and subsequently, there are more energy transactions between the CES device and the grid: the fully-competitive case requires 58\% less absolute energy traded between the CES device and the grid than the benevolent case.
\begin{figure}[b!]
\centering
\includegraphics[width=0.8\columnwidth]{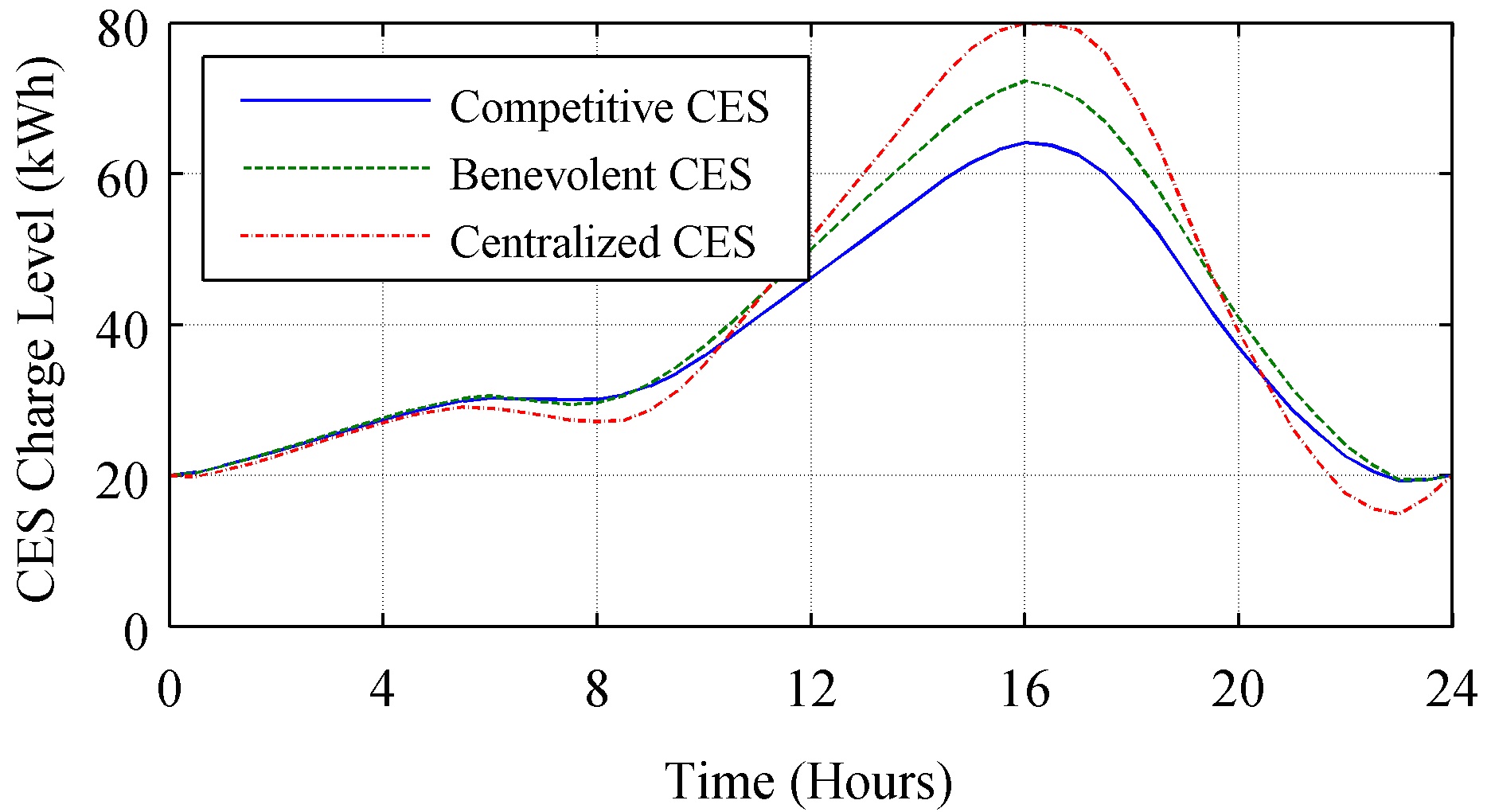}
\caption{Charge levels of the CES device for the different CES operator models with 40\% participating users.}
\label{fig:40PUCEScharge}
\end{figure}

Fig.~\ref{fig:sensitivity} shows the sensitivity of community benefits of the three systems to CES battery capacity. 
\begin{figure}[t!]
\centering
\includegraphics[width=0.82\columnwidth]{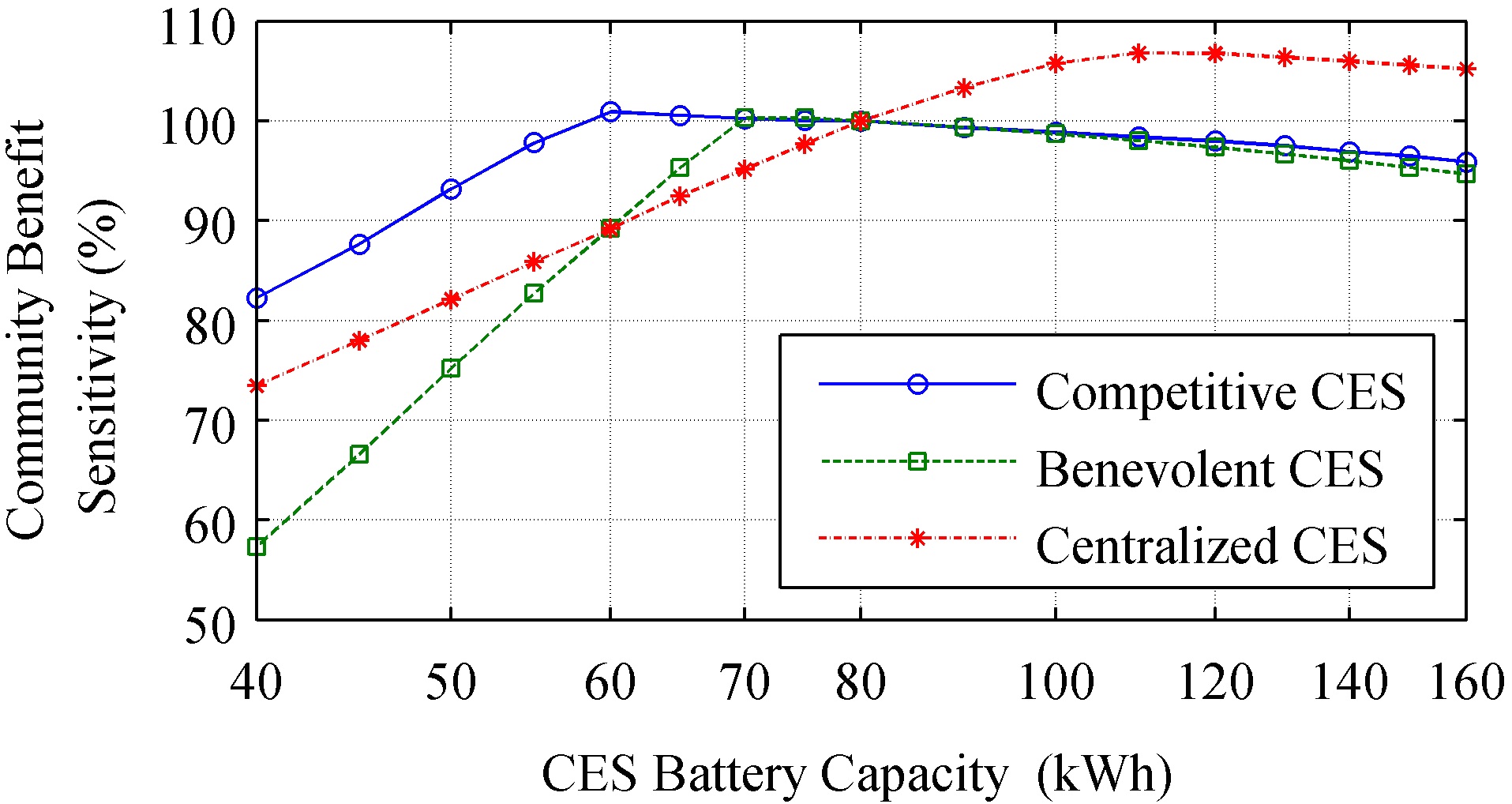}
\caption{Sensitivity of community benefit to energy storage capacity for the different CES operator models with 40\% participating users. }
\label{fig:sensitivity}
\end{figure}
The community benefit is the sum of absolute electricity cost savings of the users $\mathcal{A}\cup\mathcal{P}$ compared to the baseline and the CES revenue. As part of the CES revenue obtains from energy costs incurred by the users $\mathcal{A}$ (see \eqref{eq:id12}), the community benefit reflects the total reduction in costs paid by the community (all users and the CES operator) to the grid compared to the baseline. All models have optimal energy storage requirements corresponding to the peaks in Fig.~\ref{fig:sensitivity}. The fully-competitive CES model requires a battery capacity less than 70 kWh to provide peak community benefit compared to the other models. The centralized CES model considers load management of the entire community, not only the users $\mathcal{A}$, and therefore, requires a significantly larger storage capacity for optimal performance.

Table~\ref{tab:performance} compares the performance of the three systems over several metrics with different percentages of the users $\mathcal{A}$. 
\begin{table*}[t!]
\caption{Performance of the three CES models with different fractions of participating users (PU). PAR is peak-to-average ratio. }
\begin{center}
\begin{tabular}{ |p{2.5cm} I c | c | c I c | c | c I c | c | c I c | c | c |}
\hline
Performance Metric & \multicolumn{3}{cI}{Average PU Cost Savings (\%)} & \multicolumn{3}{cI}{\makecell*{CES Operator Revenue\\ (AU cents)}} & \multicolumn{3}{cI}{\makecell*{Community Benefit \\(AU cents)}} & \multicolumn{3}{c|}{PAR Reduction (\%)} \\\hline
PU Fraction     & 30\%   & 40\%   & 50\%   & 30\%& 40\%& 50\%& 30\%& 40\%& 50\%& 30\%  & 40\%  & 50\%    \\\hlinewd{1.2pt}
Competitive CES & 27.6 & 29.4 & 31.4 & 323 & 344 & 373 & 945 & 1023& 1123& 30.3& 31.7& 33.1 \\\hline
Benevolent CES  & 30.7 & 32.4 & 34.9 & 229 & 191 & 160 & 856 & 852 & 889 & 33.8& 35.9& 38.3 \\\hline
Centralized CES & 61.2 & 62.0 & 64.2 & -22 & -72 & -150& 1193& 1267& 1369& 37.0& 38.2& 39.5 \\\hline
\end{tabular}
\end{center}\mbox{}
\label{tab:performance}
\end{table*}
Here, the percentage cost savings and the peak-to-average ratio reductions are calculated compared to the baseline. When combined with the CES inefficiencies and price signal limitations, the CES operator revenue reduces from the fully-competitive through benevolent to centralized case in each user-percentage case. Conversely, as the CES device enacts greater demand-side management by reducing the peak-to-average ratio, the average cost saving of a user in $\mathcal{A}$ increases. The average cost saving of a user in $\mathcal{A}$ is similar under the fully-competitive and benevolent CES models, then increases notably under the centralized CES model. In fairness point of view, the users $\mathcal{A}$ enjoy greater savings than the users $\mathcal{P}$. For example, with 40\% participating users in the fully-competitive CES model, a participating user receives 29.4\% cost saving on average and a non-participating user only receives 7.92\% cost saving on average. The CES revenue is greatest for the fully-competitive CES model and decreases significantly to a loss under the centralized CES model (see Table~\ref{tab:performance}). This is because the fully-competitive model allows complete freedom to the CES operator to maximize revenue \eqref{eq:id22} while the benevolent CES operator is restricted to set a price and the centralized model eliminates the CES price signal. Overall, the total community benefit of introducing the CES device is greatest for the centralized CES followed by the fully-competitive CES and then the benevolent CES (see Table~\ref{tab:performance}). On average, the fully-competitive CES model provides 81\% of the centralized model's economic benefit compared with only 68\% for the benevolent CES model. Overall, the fully-competitive system gives the best trade-off of cost benefits between the CES operator and the users $\mathcal{A}$ of the three systems while delivering significant load leveling.

Fig.~\ref{costDistribution} depicts the distribution of cost savings for individual participating users when the fully-competitive system has 30\% participating users who have surplus energy distributions as shown in Fig.~\ref{surplusDist}. 
 \begin{figure}[t!]
\centering
\includegraphics[width=0.82\columnwidth]{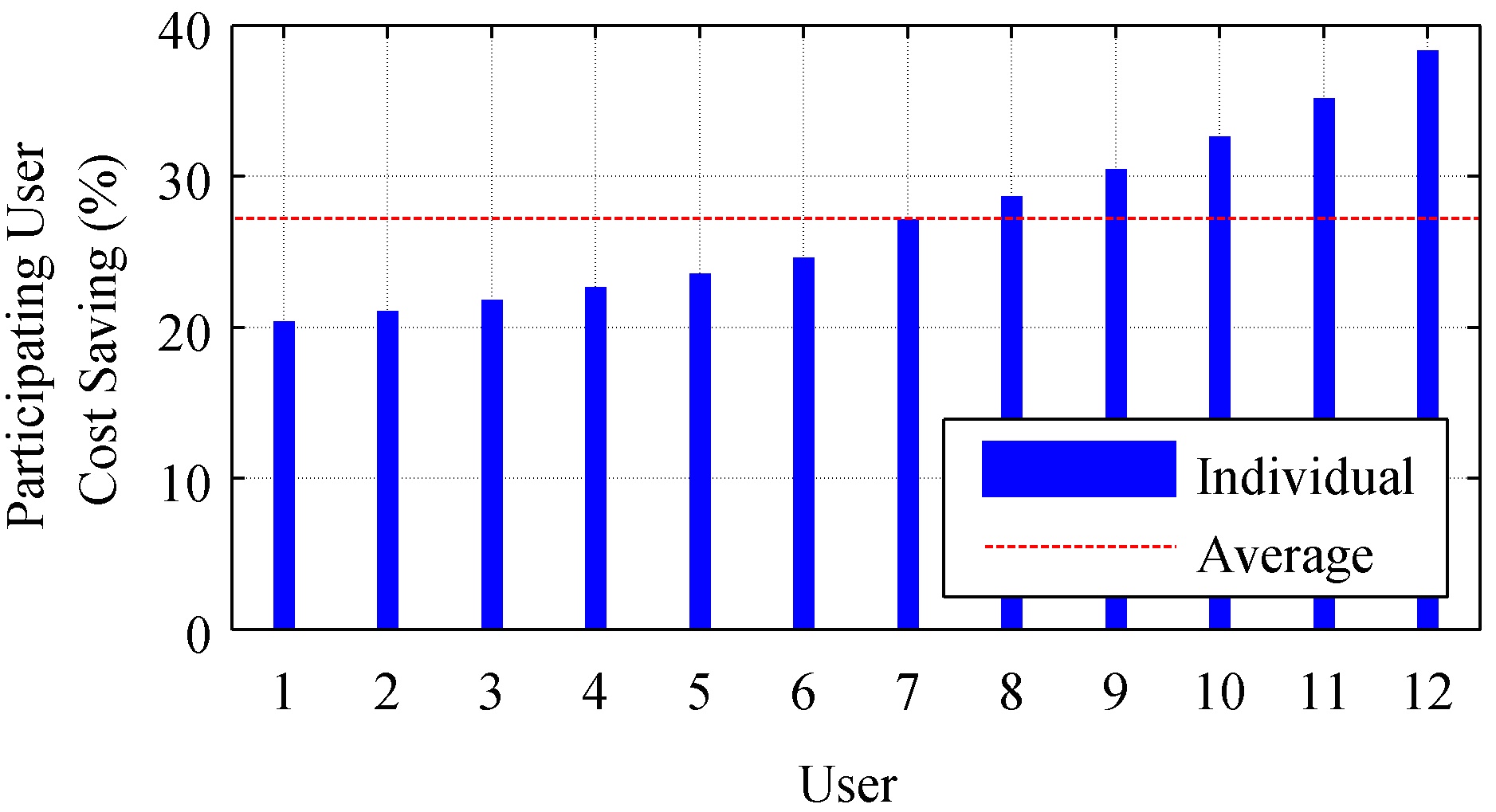}
\caption{Distribution of individual participating user cost savings with 30\% users in the fully-competitive CES model. }
\label{costDistribution}
\end{figure}
\begin{figure}[t!]
\centering
\includegraphics[width=0.82\columnwidth]{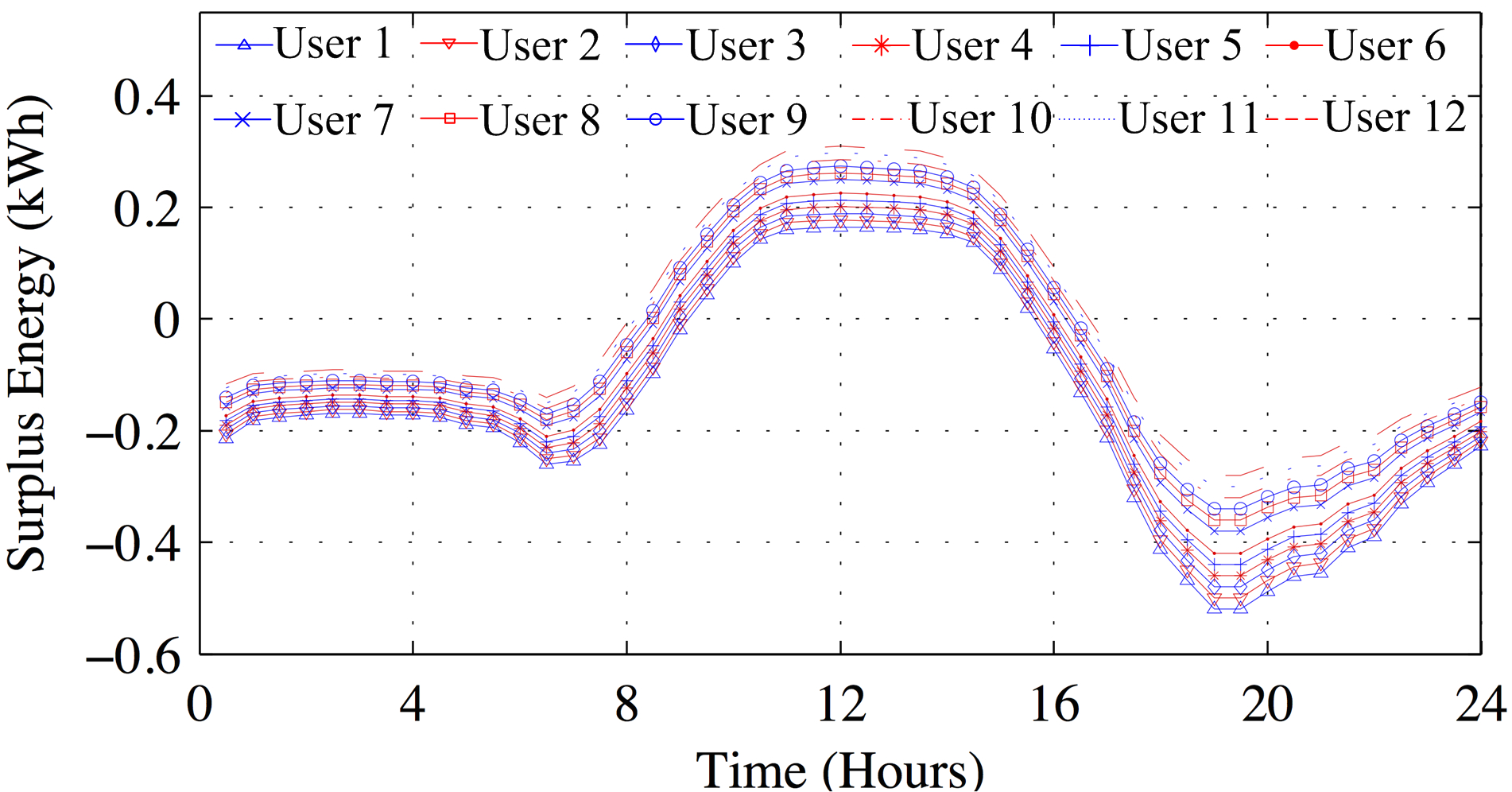}
\caption{Distribution of surplus energy of 30\% participating users in the community.}
\label{surplusDist}
\end{figure}
According to these figures, the users those who have a greater amount of surplus energy are more benefited by the system than the users with less amount of surplus energy. 

Having insights for the feasibility of the fully-competitive CES model, to investigate the effects of imperfect energy forecasts on the system, we introduce proportional variance white noise errors to the PV power and energy demand forecasts \cite{TSG1}. When averaged over a large number of simulations, the mean absolute percentage energy forecast error is equal to half of percentage white noise variance \cite{Ted}. For 40\% participating user case, when the mean absolute percentage forecast error changes from 0\% to 50\%, the average community user cost saving compared to the baseline was only reduced from 11.95\% to 11.84\%, and the average participating user saving was reduced from 29.37\% to 29.21\%. Here, for each 10\% increase in mean absolute forecast error, the average community user cost saving decreased by nearly 0.02\% while the average participating user cost saving declined by approximately 0.04\%. Similar trends were observed for both 30\% and 50\% participating user cases. Therefore, the cost benefits of the fully-competitive CES model are robust to imperfect demand and PV energy forecasts.

\section{Conclusion}\label{sec:6}
Community energy storage (CES) devices offer significant opportunities for user electricity cost savings, operator revenue, and peak-to-average ratio reduction of the grid. These benefits were shown to increase with the fraction of the participating users in the community. We have investigated three different CES operator models for community-level demand-side management and presented a fully-competitive CES operator model in a non-cooperative Stackelberg game that produces the best trade-off of operating environment between the CES operator and the users.

Interesting future work could focus on introducing more energy trading flexibility from different distributed generation sources and investigating fairness aspects between participation and non-participation in the decentralized fully-competitive energy trading system. Moreover, the proposed fully-competitive energy trading strategies in this paper could be extended to achieve socially optimal behavior of the system.


 \newcommand{\noop}[1]{}

\end{document}